\documentclass[a4paper,11pt]{article}
\usepackage[margin=1.1in]{geometry}
\usepackage{f420}

\title{Enumerating All Convex Polyhedra Glued from Squares in Polynomial Time}
\author{{\scshape Langerman} Stefan\and {\scshape Potvin} Nicolas \and {\scshape Zolotov} Boris}
\date{\today}

\begin{document} \maketitle

\begin{abstract}We present an algorithm that enumerates and classifies all edge-to-edge gluings of unit squares that correspond to convex polyhedra. We show that the number of such gluings of $n$ squares is polynomial in $n$, and the algorithm runs in time polynomial in $n$ (pseudopolynomial if $n$ is considered the only input). Our technique can be applied in several similar settings, including gluings of regular hexagons and triangles.\end{abstract}

\section{Introduction}

Given a collection of 2D polygons, a \emph{gluing} describes a closed surface by specifying how to glue each edge of these polygons onto another edge. Alexandrov's uniqueness theorem~\cite{alex} states that any valid gluing that is homeomorphic to a sphere and that does not yield a total facial angle greater than $2\pi$ at any point, corresponds to the surface of a unique convex 3D polyhedron (doubly covered convex polygons are also regarded as polyhedra). Note that the original polygonal pieces might need to be folded to obtain this 3D surface.

There is no known exact algorithm for reconstructing the 3D polyhedron~\cite{bannister2014galois,kpd09-approx}. Enumerating all possible valid gluings is also not an easy task, as the number of gluings can be exponential even for a single polygon~\cite{DDLO02}. Complete enumerations of gluings and the resulting polyhedra are only known for very specific cases such as the Latin cross~\cite{ddlop99}, a single regular convex polygon~\cite{DO07}, and a collection of regular pentagons~\cite{alz-penta}.

The case when the polygons to be glued together are all identical regular $k$-gons, and the gluing is \emph{edge-to-edge} was studied recently for $k \ge 6$~\cite{kl17-hex}. The aim of this paper is to study the case of $k=4$: namely, to {\it enumerate} all valid gluings of squares and {\it classify} them up to isomorphism.

\section{Chen—Han algorithm for gluings of squares}

In~\cite{DO07} it is shown that polyhedra are isomorphic if the lengths of shortest geodesic paths between their vertices of nonzero curvature coincide. Thus, the problem of finding out if two gluings are isomorphic can be reduced to finding out the geodesic distances between vertices of a gluing. Algorithm we are using for this is the Chen—Han algorithm~\cite{chen-han}.

The idea of the algorithm is to project a cone of all possible paths from the source onto the surface of the gluing. This algorithm runs in $O(n^2)$ time. To apply it for arbitrary edge-to-edge gluings of squares, it has to be proven that the running time is preserved. To do this, we prove the following Theorem.

\begin{theorem} \label{thm:shortestSquare}
	If $T$ is a square of the gluing and $\pi$ is a geodesic shortest path between two vertices of the gluing then the intersection between $\pi$ and $T$ is of at most 5 segments.
\end{theorem}

To prove this Theorem, we need an additional definition and a series of lemmas.

\begin{definition}
	Let $a_ib_i$ be a segment of the intersection between $\pi$ and $T$ with $a_i, b_i \in \partial T$. The points $a_i$, $b_i$ divide the boundary of $T$ into two parts. Let $n_1$ and $n_2$ be the numbers of vertices of $T$ in the parts respectively. Then the segment $a_ib_i$ has \emph{type $s$} if $\min (n_1, n_2) = s$. If $T$ is a convex $r$-gon then the possible types of a segment can range between 1 and $\lfloor r / 2 \rfloor$ (see Figure~\ref{fig:segmType}).
\end{definition}

\begin{figure}[h]
	\input{img/past-and-future-circles}
\end{figure}

\begin{lemma}
\label{lm:pastFuture}
	Let the intersection between $\pi$ and $T$ consist of $m$ segments, denote them by $a_1b_1 \ldots a_mb_m$ according to the order in which they appear in $\pi$. If $a_ib_i$ is a segment of the intersection between $\pi$ and $T$, then \begin{enumerate}
	\item \label{item:past} no point of $a_1b_1, \ldots, a_{i-1}b_{i-1}$ may lie in the disk centered at $b_i$ with radius $|a_ib_i|$ (shown in Figure~\ref{fig:past}),
	\item \label{item:future} no point of $a_{i+1}b_{i+1}, \ldots, a_mb_m$ may lie in the disk centered at $a_i$ with radius $|a_ib_i|$ (shown in Figure~\ref{fig:future}),
	\item \label{item:noother} no point of $a_jb_j$, $j \ne i$ may lie inside the disk having $a_ib_i$ as its diameter.
\end{enumerate}
\end{lemma}

\begin{proof} Let us prove item~(\ref{item:past}). If there is a point $p$ of a segment $a_jb_j$ preceding $a_ib_i$ inside that circle then (because $T$ is convex) the segment $pb_i$ lies entirely inside $T$ and is shorter than $a_ib_i$. Then we can replace the path
	\[\pi = \ldots a_jpb_j \ldots a_ib_i \ldots\]
with $\ldots a_jpb_i \ldots$, which is shorter than $\pi$, thus $\pi$ is not the shortest.

	The proof of~(\ref{item:future}) is analogous. To prove~(\ref{item:noother}) let us note that the intersection between the disks from~(\ref{item:past}) and~(\ref{item:future}) covers the smaller circle we are considering (see Figure~\ref{fig:typeOne}), thus no segments either from the past or the future of path $\pi$ can lie inside it.\end{proof}

\begin{figure}[h]
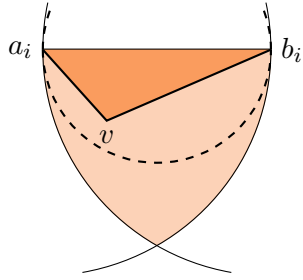

	\centering
	\tikz{
		\filldraw [draw=fill1,fill=fill1,opacity=0.45]
			(-1.5,0) arc(180:240:3cm) arc(300:360:3cm) -- cycle;
		\draw (-1.5,0) arc (180:260:3cm); \draw (1.5,0) arc (360:280:3cm);
		\draw (-1.5,0) arc (180:168:3cm); \draw (1.5,0) arc (360:372:3cm);
		\draw[thick] (-1.5,0) -- (1.5,0);
		\draw (-1.5,0) node[left]{$a_i$}; \draw (1.5,0) node[right]{$b_i$};
		\draw[thick,dashed] (-1.5,0) arc(180:360:1.5cm);
			\draw[thick,dashed] (-1.5,0) arc(180:156:1.5cm);
			\draw[thick,dashed] (1.5,0) arc(360:384:1.5cm);
		\filldraw[thick,draw=black,fill=fill1] (-1.5,0) -- (235:1.15cm) -- (1.5,0);
		\draw (235:1.15cm) node[below]{$v$};
	}
	\caption{No points of path $\pi$ lie inside the disk
		whose diameter is $a_ib_i$ and thus, in turn,
		inside the triangle $a_ivb_i$.}
	\label{fig:typeOne}
\end{figure}

\begin{lemma}
\label{lm:typeOne}
	If the angle at a vertex $v$ of polygon $T$ is not less than $90^\circ$ then there is at most one segment $a_ib_i$ of the intersection $\pi \cap T$ of type 1 such that the vertex $v$ lies between points $a_i$ and $b_i$.
\end{lemma}

\begin{proof} Suppose there are multiple segments satisfying the statement of the lemma. Let us denote by $a_ib_i$ the one whose end is the farthest from $v$. Due to Lemma~\ref{lm:pastFuture} all the points of the segments $a_jb_j$, $j \ne i$ of the intersection $\pi \cap T$ must lie outside the disk $D$ whose diameter is $a_ib_i$.
	
	This disk covers all the points $s$ such that $\measuredangle a_isb_i \ge 90^\circ$, thus it covers the triangle $a_ivb_i$ (see Figure~\ref{fig:typeOne}) and consequently no other points of the path $\pi$ can lie inside this triangle.

	All the other segments satisfying the statement have at least some of their points lying inside $a_ivb_i$. Thus we arrive to a contradiction: $\pi$ is not the shortest path.\end{proof}

\begin{lemma}
\label{lm:squareTypeTwo}
	Let $T = ABCD$ be a square, then there is at most one segment
	of type 2 in $\pi \cap T$.
\end{lemma}

\begin{proof} A segment of type 2 connects either $AB$ with $CD$ or $BC$ with $DA$. Since $\pi$ is a shortest path, segments that are parts of $\pi$ can not intersect each other. Therefore, without loss of generality all the segments of type 2 connect $AB$ with $CD$.
	
	Among all these segments let us consider the segment $a_ib_i$ of type 2 with the greatest possible~$i$. Without loss of generality $a_i \in AB$, $b_i \in CD$. The disk centered at $b_i$ with radius $|a_ib_i| \ge |CD|$ covers $CD$, thus, by Lemma~\ref{lm:pastFuture}, there are no segments $a_jb_j$, $j<i$ with endpoints on $CD$. It means that $a_ib_i$ is the only segment of type 2, which completes the proof.\end{proof}

\begin{proof}[Proof of Theorem~\ref{thm:shortestSquare}]
	A segment in the intersection can either have type 1 or type 2. By Lemma~\ref{lm:typeOne}, there are at most 4 segments of type 1: at most one for each vertex. By Lemma~\ref{lm:squareTypeTwo}, there is at most one segment of type 2. This sums up to 5 segments.\end{proof}

Theorem~\ref{thm:shortestSquare} implies the following Theorem.

\begin{theorem} \label{thm:chruntime}
	The isomorphism between two edge-to-edge gluings of at most $n$ squares can be tested in $O(n^2)$ time.
\end{theorem}

\section{Bounds on the number of egde-to-edge gluings of squares}

In this section, we prove that the number of edge-to-edge gluings of $n$ squares is polynomial in $n$. This result allows to develop a polynomial algorithm to list all the gluings.

\begin{theorem} \label{thm:n36}
	There are $O \left( n^{36} \right)$ edge-to-edge gluings of at most $n$ squares that correspond to convex polyhedra.
\end{theorem}

\begin{proof}
	Triangulate the polyhedron corresponding to the gluing and draw its faces on the square grid. By Gauss—Bonnet theorem, the polyhedron has no more than 8 vertices, and thus at most 18 edges. An edge shared by two faces must have the same lengths of $x$- and $y$-projections on the drawings of these faces, see Figures~\ref{fig:emA},~\ref{fig:emB}.

\begin{figure}[h]
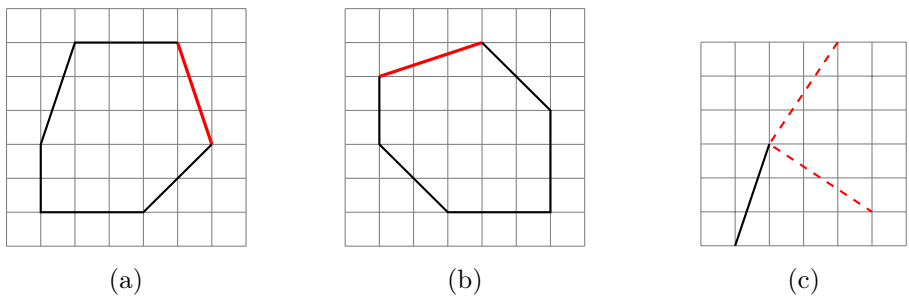
 \centering
\begin{subfigure}[t]{3.6cm} \centering
\tikz[scale=0.45]{
	\foreach \i in {-1,...,6} {
		\draw[gray] (-1,\i) -- (6,\i) (\i,-1) -- (\i,6);
	}
	\draw[very thick,red] (4,5) -- (5,2);
	\draw[thick] (5,2) -- ++(-2,-2) -- ++(-3,0) --
		++(0,2) -- ++(1,3) -- ++(3,0);
} \caption{} \label{fig:emA} \end{subfigure} \hspace{0.6cm}
\begin{subfigure}[t]{3.6cm} \centering
\tikz[scale=0.45]{
	\foreach \i in {-1,...,6} {
		\draw[gray] (-1,\i) -- (6,\i) (\i,-1) -- (\i,6);
	}
	\draw[very thick,red] (0,4) -- (3,5);
	\draw[thick] (3,5) -- ++(2,-2) -- ++(0,-3) --
		++(-3,0) -- ++(-2,2) -- ++(0,2);
} \caption{} \label{fig:emB} \end{subfigure} \hspace{0.8cm}
\begin{subfigure}[t]{3.2cm} \centering
\tikz[scale=0.45]{
	\foreach \i in {0,...,6} {
		\draw[gray] (0,\i) -- (6,\i) (\i,0) -- (\i,6);
	}
	\draw[thick] (1,0)--(2,3);
	\draw[thick,red,dashed] (4,6)--(2,3)--(5,1);
} \caption{} \label{fig:emC} \end{subfigure}

\caption{(a), (b) Highlighted edge has the same lengths of projections on
	the drawings of two faces. (c) Two ways to place an edge with given
	projections that preserve convexity of the face.}
\end{figure}

Count the number of sets of triangles satisfying this restriction and taking up at most $n$ squares. For each edge, we can pick lengths of its $x$- and $y$-projections. Since the edge is a part of a flat face, all the squares that intersect the edge are distinct. There is at most $n$ of them, which yields that both projections are at most $n$, so there is at most $n^2$ ways to choose the edge.

Once the projections of the edges are known, let us draw the faces on the grid. At every vertex, there is at most two ways to place the next edge, such that the convexity of the face is preserved, those differ by $\frac{\pi}{2}$, see Figure~\ref{fig:emC} for an example. This adds at most $2^{2 \cdot 18}$ ways to draw the faces once the edges are known, which gives the total of at most $(2n)^{36}$ gluings.
\end{proof}

\begin{theorem} \label{thm:n3}
	There are $\Omega \left( n^3 \right)$ edge-to-edge gluings of at most $n$ squares that correspond to convex polyhedra.
\end{theorem}

\begin{proof}
To prove the theorem, we construct a series of such gluings. These gluings correspond to doubly-covered octagons, the octagons being obtained by cutting edges of a rectangle with sides no longer than $\frac{\sqrt{n}}{2}$, one at least twice as long as the other, see Figure~\ref{fig:cutEx}.

\begin{figure}[h]
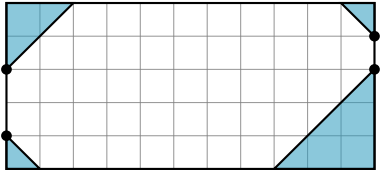
 \centering
\tikz[scale=0.44]{
	\foreach \x in {0,...,11} {\draw[gray,opacity=0.8] (\x,0) -- (\x,5);}
	\foreach \y in {0,...,5} {\draw[gray,opacity=0.8] (0,\y) -- (11,\y);}
	\draw[thick] (0,0) -- (11,0) -- (11,5) -- (0,5) -- cycle;
	\filldraw[fill=bluetri,thick,fill opacity=0.55] (10,5) -- (11,4) -- (11,5) -- cycle;
	\filldraw[fill=bluetri,thick,fill opacity=0.55] (8,0) -- (11,3) -- (11,0) -- cycle;
	\filldraw[fill=bluetri,thick,fill opacity=0.55] (0,0) -- (1,0) -- (0,1) -- cycle;
	\filldraw[fill=bluetri,thick,fill opacity=0.55] (0,3) -- (2,5) -- (0,5) -- cycle;
	\fill (11,4) circle[radius = 1.6mm]
		(11,3) circle[radius = 1.6mm]
		(0,3) circle[radius = 1.6mm]
		(0,1) circle[radius = 1.6mm];}
\caption{An example of an octagon produced by cutting angles of a rectangle} \label{fig:cutEx}
\end{figure}

Pick width and height of the rectangle. Denote width by $a$, pick it so that $0 < a \le \frac12 \sqrt{n}$. Denote height by $b$, pick it so that $0 < b \le \frac{a}2$. An octagon is defined by its vertices on the vertical edges of the rectangle (those are highlighted in Figure~\ref{fig:cutEx}). There are $\binom{b}{2}^2 \sim \frac14 b^4$ ways to choose these vertices.

To sum up, the number of octagons obtained by cutting edges of a rectangle is
\[ \Omega\ll \sum\limits_{a=1}^{\frac{\sqrt{n}}{2}} \sum\limits_{b=1}^{\frac{a}{2}} b^4 \rr =
	\Omega\ll \sum\limits_{a=1}^{\frac{\sqrt{n}}{2}} a^5 \rr = \Omega \ll n^3 \rr. \]

\end{proof}

\begin{theorem}
	The bound of Theorem~\ref{thm:n3} is tight: there are $O(n^3)$ doubly covered convex polygons that can be glued from $n$~squares.
\end{theorem}

\begin{proof} Edges of a doubly covered polygon glued from squares can only have four directions: vertical, horizontal, or inclined by $\frac{\pi}{2}$ or $\frac{3\pi}{2}$. Thus any doubly covered polygon glued from squares is an octagon cut from a rectangle. Some edges of the octagon, however, may have zero length.

\begin{figure}[h]
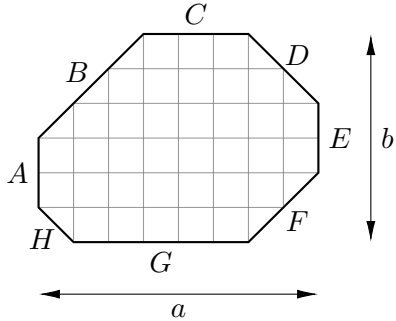
 \centering \tikz[scale=0.46]{
	\foreach \x in {1,...,7} {\draw[gray,opacity=0.8] (\x,0) -- (\x,6);}
	\foreach \y in {1,...,5} {\draw[gray,opacity=0.8] (0,\y) -- (8,\y);}
	\filldraw[fill=white,draw=white,thick] (0,0) -- (0,1) -- (1,0) -- cycle;
	\filldraw[fill=white,draw=white,thick] (6,0) -- (8,0) -- (8,2) -- cycle;
	\filldraw[fill=white,draw=white,thick] (8,4) -- (8,6) -- (6,6) -- cycle;
	\filldraw[fill=white,draw=white,thick] (0,3) -- (3,6) -- (0,6) -- cycle;
	\draw[thick] (1,0) -- ++(5,0) -- ++(2,2) -- ++(0,2)
		-- ++(-2,2) -- ++(-3,0) -- ++(-3,-3) -- ++(0,-2) -- cycle;
	\draw (3.5,0) node[below]{$G$} (7.4,0.6) node{$F$}
		(8,3) node[right]{$E$} (7.4,5.4) node{$D$}
		(4.5,6) node[above]{$C$} (1.1,4.9) node{$B$}
		(0,2) node[left]{$A$} (0.1,0.1) node{$H$};
	\draw[<->] (0,-1.5) -- (8,-1.5);
	\draw[<->] (9.5,0) -- (9.5,6);
	\draw[white] (-1.5,0) -- (-1.5,6) (-1.5,3) node[left]{$b$};
	\draw (4,-1.5) node[below]{$a$} (9.5,3) node[right]{$b$};
}
	\caption{An octagon cut from rectangle and its dimensions}
	\label{fig:octparam}
\end{figure}

Let us denote the numbers of squares traversed by the edges by $A$,~\ldots, $H$, as shown in Figure~\ref{fig:octparam}. The width of the circumscribed rectangle is denoted by $a$, and the height of the circumscribed rectangle is denoted by $b$. We can think that $a \le b$. Note that

\[ \begin{cases}
	A+B+H = D+E+F = b \\
	B+C+D = F+G+H = a \\
	b \le a\\ a \cdot b \le \frac{n}{2}
\end{cases} \]

Six variables out of these ten: $a$, $A$, $B$, $D$, $F$, $H$ — define other four. This means, given values of $a$, $A$, $B$, $D$, $F$, $H$, the octagon is either defined uniquely or non-existent:

\[ \begin{array}{l}
	b = A+B+H \\
	E = A+B+H-D-F \\
	C = a-B-D \\
	G = a-H-F \\
\end{array} \]

Let us now count the ways to pick $a$, $A$, $B$, $D$, $F$, $H$.

{\bfseries First case} \( 1 \le b \le a \le \sqrt{n} \). In this case all the variables $A$,~\ldots, $H$ are at most $\sqrt{n}$, which yields the number of ways to pick six variables does not exceed $(\sqrt{n})^6 = n^3$.

{\bfseries Second case} \(a \ge \sqrt{n}\), \(b \le \frac{n}{a}\). In this case $A$, $B$, $D$, $F$, $H$ are at most $\frac{n}{a}$ since they all contribute to the height of the octagon. The number of ways to pick six variables is hence equal to
	\[ \slim_{a=\sqrt{n}}^n \ll \frac{n}{a} \rr^5.\]

We now split and estimate this sum. Assume $n$ is a power of $2$ or consider the closest to $n$ power of $2$ from above.\vspace{-0.7cm}

\begin{align*}
	& \slim_{a=\sqrt{n}}^n \ll \frac{n}{a} \rr^5\ \ \le\ \ \sloglim
		\slim_{a=2^{i-1}}^{2^i} \ll \frac{n}{a} \rr^5\ \ \le \\
 	\le\ \ & \sloglim \ll \frac{n}{2^{i-1}} \rr^5 \cdot 2^i\ \ =\ \ \sloglim
 		2n \cdot \ll \frac{n}{2^{i-1}} \rr^4\ \ = \\
	=\ \ & \slim_{i=0}^{\frac{\log n}{2}} 2n \cdot \ll \frac{\sqrt{n}}{2^{i-1}} \rr^4
		\ \ =\ \ 2n^3 \cdot \slim_{i=0}^{\frac{\log n}{2}} \frac{16}{2^{4i}}\ \ =\ \ O(n^3).
\end{align*} \end{proof}

We implemented an algorithm that enumerates all the gluings of at most $n$ squares for a given graph structure of a convex polyhedron. It showed that one gluing can admit several ways to cut itself into flat polygons, see Figure~\ref{fig:twonets}. Thus it can appear in the list several times.

\begin{figure}[h]
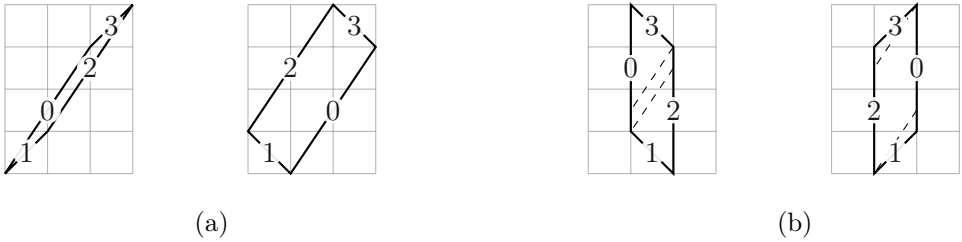
 \centering
      \begin{subfigure}[b]{0.44\textwidth} \centering
            \tikz[scale=0.56]{\fill[white] (-0.4,-0.4) rectangle (4.4,4.4);
                  \foreach \i in {0,...,3} {\draw[gray,opacity=0.6] (\i,0) -- (\i,4);}
                  \foreach \i in {0,...,4} {\draw[gray,opacity=0.6] (0,\i) -- (3,\i);}
                  \draw[thick] (2,3) -- (0,0) \midnode{0} --
                        (1,1) \midnode{1} -- (3,4) \midnode{2} --
                        cycle \midnode{3};}
            \quad
            \tikz[scale=0.56]{\fill[white] (-0.4,-0.4) rectangle (4.4,4.4);
                  \foreach \i in {0,...,3} {\draw[gray,opacity=0.6] (\i,0) -- (\i,4);}
                  \foreach \i in {0,...,4} {\draw[gray,opacity=0.6] (0,\i) -- (3,\i);}
                  \draw[thick] (3,3) -- (2,4) \midnode{3} --
                        (0,1) \midnode{2} -- (1,0) \midnode{1} --
                        cycle \midnode{0};}
      \caption{} \end{subfigure} \qquad
      \begin{subfigure}[b]{0.44\textwidth} \centering
            \tikz[scale=0.56]{\fill[white] (-0.4,-0.4) rectangle (4.4,4.4);
                  \foreach \i in {0,...,3} {\draw[gray,opacity=0.6] (\i,0) -- (\i,4);}
                  \foreach \i in {0,...,4} {\draw[gray,opacity=0.6] (0,\i) -- (3,\i);}
                  \draw[dashed] (2,3)--(1,1.5) (2,2.5)--(1,1);
                  \draw[thick] (1,4) -- (1,1) \midnode{0} --
                        (2,0) \midnode{1} -- (2,3) \midnode{2} --
                        cycle \midnode{3};}
            \quad
            \tikz[scale=0.56]{\fill[white] (-0.4,-0.4) rectangle (4.4,4.4);
                  \foreach \i in {0,...,3} {\draw[gray,opacity=0.6] (\i,0) -- (\i,4);}
                  \foreach \i in {0,...,4} {\draw[gray,opacity=0.6] (0,\i) -- (3,\i);}
                  \draw[dashed] (2,4)--(1,2.5) (2,1.5)--(1,0);
                  \draw[thick] (2,4) -- (2,1) \midnode{0} --
                        (1,0) \midnode{1} -- (1,3) \midnode{2} --
                        cycle \midnode{3};}
      \caption{} \end{subfigure}
      \caption{Doubly covered parallelogram can be cut into two flat quadrilaterals in two ways, the latter consisting of its faces}
      \label{fig:twonets}
\end{figure}

\section{Algorithm to classify edge-to-edge gluings of squares}

The algorithm consists of the following steps:

\begin{enumerate}
	\item Generate the list of all edge-to-edge gluings of at most $n$ squares,
	denote it $L(n)$. Due to Theorem~\ref{thm:n36}, this step takes polynomial time.
	\item For each gluing in $L(n)$, generate matrix of pairwise distances
	between its vertices. Due to Theorem~\ref{thm:chruntime},
	this step takes $O(n^3)$ time per gluing.
	\item Unicalize the list of matrices up to homothety and permutation of rows and columns, leave only corresponding elements of $L(n)$. Since the matrices are of at most 8 rows and 8 columns, it takes polynomial time to remove duplicates from the list.
\end{enumerate}

The output of this algorithm is the list of all non-isomorphic edge-to-edge gluings of at most $n$ squares.

\section{Discussion}

The cornerstone of the technique we have been using is the possibility to draw a face of a polyhedron glued from squares on a planar grid. It allows us to estimate the number of valid gluings. The same technique can seemingly be applied for the cases of regular hexagons and triangles, since these polygons also tile the plane.

The method we used to prove the bound on the number of intersections between a shortest path and a polygon can also be generalized to arbitrary polygons.

\section*{Acknowledgements}

S.~L. is Directeur de recherches du F.R.S.-FNRS.

B.~Z. is supported in part by the Foundation for the Advancement of Theoretical Physics and Mathematics «BASIS» and in part by Russian Foundation of Basic Research (grant 20-01-00488).

\bibliography{boris-bac}{}
\bibliographystyle{plain}

\end{document}